\documentclass[a4paper,11pt]{article}

\usepackage{amsmath}
\usepackage[table]{xcolor}
\usepackage[english]{babel}
\usepackage{a4wide}
\usepackage[utf8]{inputenc}
\usepackage{amsmath}
\usepackage{amssymb}
\usepackage{amsthm}
\usepackage{amsfonts}
\usepackage{amsmath}
\usepackage{latexsym}
\usepackage{ifthen, enumitem, mathrsfs, euscript,anyfontsize}
\usepackage{graphicx}
\usepackage[table]{xcolor}
\usepackage{listings}
\usepackage{url}
\usepackage{enumerate}
\usepackage{diagbox}
\usepackage{enumitem}
\usepackage{svg}
\usepackage[ruled]{algorithm2e}
\usepackage{subfig}
\usepackage{multicol}

\newtheorem{theorem}{Theorem}

\newtheorem{lemma}{Lemma}

\newtheorem{conjecture}{Conjecture}
\newtheorem*{definition}{Definition}

\newcommand{\note}[1]{\color{blue}{#1}\color{black}}

\newcommand{\at}[2]{#1[#2]}
\def\icr{ICR}
\def\Hi{H}
\def\Di{D}
\def\Pa{P}

\def\rep{rep}
\def\orbit{orbit}
\def\nxt{R}
\def\parent{parent}
\def\dif{\Delta}

\def\disc{discrepancy}

\definecolor{codegreen}{rgb}{0,0.6,0}
\definecolor{codegray}{rgb}{0.5,0.5,0.5}
\definecolor{codepurple}{rgb}{0.58,0,0.82}
\definecolor{backcolour}{rgb}{1.0,1.0,1.0}

\lstdefinestyle{mystyle}{
    backgroundcolor=\color{backcolour},
    commentstyle=\color{codegreen},
    keywordstyle=\color{magenta},
    numberstyle=\tiny\color{codegray},
    stringstyle=\color{codepurple},
    basicstyle=\ttfamily\scriptsize,
    breakatwhitespace=false,
    keepspaces,
    breaklines=true,
    captionpos=b,
    keepspaces=true,
    numbers=none,
    numbersep=5pt,
    showspaces=false,
    showstringspaces=false,
    showtabs=false,
    columns=fullflexible,
    tabsize=2
}

\DeclareMathOperator{\reps}{\mathbf{Reps}}

\begin{document}

\title{
De Bruijn Sequences with Minimum Discrepancy}

\author{
\small Nicolás Álvarez 
\and
\small
Verónica Becher 
\and 
\small Martín Mereb 
\and 
\small Ivo Pajor 
\and 
\small Carlos Miguel Soto
}

\maketitle

\begin{abstract}
The discrepancy of a binary string is the maximum (absolute) difference between the number
of ones and the number of zeroes over all possible substrings of the given binary string.
In this note we determine the minimal discrepancy that a binary de Bruijn sequence of order $n$ can achieve, which is $n$.
This was an open problem until now.
We give an algorithm that constructs a binary de Bruijn sequence with  minimal discrepancy.
A slight modification of this algorithm deals with arbitrary alphabets and yields de Bruijn sequences of order $n$  with discrepancy at most $1$ above the trivial lower bound $n$.
\end{abstract}

{\small \tableofcontents}

\section{Statement of Results}

A de Bruijn sequence  of order $n$ over a $b$-symbol alphabet  is a circular string of length $b^n$ such that every string of length $n$ occurs exactly once in it.
The discrepancy of a  given string  is a non negative integer that indicates the maximum difference, in any substring, between the number of occurrences of the  most occurring symbol 
and the least  occurring symbol in that substring.
Formally:  
let $\Sigma$ be an alphabet and $w$ a string over $\Sigma$. Let $|w|_a$ denote the number of occurrences of the symbol $a$ in $w$. 
Let $S(w)$ denote the set of all substrings of $w$ where we interpret $w$ as a circular string.
The discrepancy of a string $w$ over alphabet $\Sigma$  is defined as
\[
\disc(w) =
\underset{s\in S(w)}{\max} \left(
\underset{a\in \Sigma}{\max}\ |s|_a - \underset{c\in \Sigma}{\min}\ |s|_c\right).
\]

The maximum and minimum discrepancies attainable by a de Bruijn sequence of order~$n$ are $\Theta(2^n / \sqrt{n})$ and $\Theta(n)$ respectively. This was proved by Gabric and Sawada in 
\cite[Theorem 5.3]{Sawada} and~\cite[Theorem 2.2]{Sawada}.

What  is the exact  minimum discrepancy that a de Bruijn sequence can achieve?  
Clearly, in every alphabet,  the discrepancy of a de Bruijn sequence of order~$n$ is at least~$n$, because there must exist a run of 
$n$ consecutive equal symbols.
Here we show that for the binary alphabet the minimum discrepancy achievable by a de Bruijn sequence of order~$n$ is exactly~$n$.
We give an  efficient algorithm to compute it. For alphabets with more than two symbols, our algorithm constructs a de Bruijn sequence of order $n$ with discrepancy at most $n+1$.

\begin{theorem}\label{thm:main}
There is an algorithm  that produces a de Bruijn sequence of order $n$  with discrepancy $n$ in case the alphabet has two symbols, and with discrepancy at most $n+1$ in case the alphabet has more than two symbols. The algorithm computes in $O(n)$ memory and it outputs each symbol in $O(n)$ time.
\end{theorem}

Thus, in the case of the binary alphabet our algorithm yields minimal discrepancy. 
For larger alphabets we experimentally found that for some values of $n$ there exists a de Bruijn sequence of order $n$ with discrepancy $n$, but not for others.

Gabric and Sawada in~\cite{Sawada} ask whether 
Huang's~\cite{Huang} algorithm  produces binary de Bruijn sequences with minimum discrepancy.
The answer is no, but a variant of it does: our algorithm is a variant of Huang's algorithm, twisted to achieve minimum discrepancy.
Our algorithm works for any alphabet, thereby, it solves the problem posed by Gabric and Sawada of determining whether the discrepancy bounds for the binary alphabet hold for arbitrary alphabets  as well.

\begin{section}{Structure of the Proof of Theorem 1}
    The structure of the algorithm and its correctness proof is as follows:
    \begin{itemize}
        \item First, we translate the problem to the language of de Bruijn graphs, using the well known result that Hamiltonian cycles in the de Bruijn graph correspond to de Bruijn sequences~\cite{dB}. 
        
        \item We define a 
        subgraph of the de Bruijn graph, called the valid subgraph, where all arcs are of a certain form and every path in that graph corresponds to a string with bounded discrepancy.
       
        \item We decompose the de Bruijn graph into node-disjoint cycles using the incremented cycle register rule, as is done by Huang~\cite{Huang}. We prove that these cycles are all contained within the valid subgraph.
        
        \item We choose a way to connect these cycles in a tree-like structure to form a Hamiltonian cycle, such that the arcs between the cycles are also contained in the valid subgraph.

        \item Finally, we develop an algorithm to traverse this tree in a depth first manner, completely traversing all cycles and thus producing a Hamiltonian cycle. Since it is contained within the valid subgraph, this guarantees that the resulting Hamiltonian cycle has bounded discrepancy.
    \end{itemize}
\end{section}

\section{De Bruijn Graph and Discrepancy}
Given an integer $n$ and an alphabet $\Sigma = \{0,1,\dots,k-1\}$, where $k$ is a positive integer, the de Bruijn graph $B_n = (V_n, A_n)$ is a graph with node set $V_n = \Sigma^n$, the set of all strings of length $n$, and an arc $(s, t) \in A_n$ if and only if the suffix of $s$ of length $n-1$ equals the prefix of $t$ of length $n-1$. The alphabet is interpreted modulo $k$. For example, for $k=3$ we take $1+2 = 0$.

Hamiltonian cycles in the de Bruijn graph correspond to de Bruijn sequences:
if $s_0, s_1, \dots, s_{k-1}$ is a Hamiltonian cycle in $B_n$, then the sequence
\[
\at{s_0}{0}, \at{s_1}{0}, \dots, \at{s_{k-1}}{0}
\]
consisting of the first symbol of each string is a de Bruijn sequence. Thus, each string of length $n$ appears exactly once in the sequence when viewed cyclically.

\begin{definition}[Difference]
    The discrepancy of a string $w$ is defined as, $\disc:\Sigma^*\to \mathbb{Z} $
    \[
        \disc(w) =
        \underset{s\in S(w)}{\max} \left(
            \underset{a\in \Sigma}{\max}\ |s|_a - \underset{c\in \Sigma}{\min}\ |s|_c\right).
            \] 
    where the value $\underset{a\in \Sigma}{\max}\ |s|_a - \underset{c\in \Sigma}{\min}\ |s|_c$ is the \emph{difference} of the substring $s$, denoted $\Di(s)$.
\end{definition}
            
In order to bound the discrepancy of de Bruijn sequences, we need to bound the difference of all its substrings. A substring of a de Bruijn sequence is a sequence of the form
\[
\at{s_{i}}{0}, \at{s_{i+1}}{0}, \dots, \at{s_{j}}{0}
\]
where $s_i, s_{i+1}, \dots, s_j$ is a contiguous subsequence of the Hamiltonian cycle. Therefore, substrings of a de Bruijn sequence correspond to paths in the de Bruijn graph.

\section{Valid Subgraph}
\subsection{Histograms}

Let us define a concept that will help us to calculate the difference of a string.

\begin{definition}[Histogram]
Given a string $s \in \Sigma^*$ let $\Hi(s) : \Sigma \to \mathbb{Z}$ be defined as $\Hi(s)(c) = |s|_c$.
\end{definition}

The difference of a histogram is defined as follows.

\begin{definition}[Difference of a histogram]
The difference $\Di(\Hi)$ of a histogram $\Hi : \Sigma \to \mathbb{Z}$ is given by
\[\Di(\Hi) = \max_{i,j \in \Sigma} \Hi(i) - \Hi(j).
\]
\end{definition}

As expected, $\Di(\Hi(s)) = \Di(s)$. We use the same symbol to denote the difference of a string and the difference of a histogram. However, the meaning of the symbol will be clear from the context in which it is used. Another operation that we need on histograms is the \emph{partial sum}.

\begin{definition}[Partial sum of a histogram]
    Given $\Hi : \Sigma \to \mathbb{Z}$, define its partial sum $\Pa(\Hi): \Sigma \to \mathbb{Z}$ as 
    \[\Pa(\Hi)(i) = \sum_{j = 0}^{i} \Hi(j).
    \]
\end{definition}

Let $K$ be the constant one histogram. 
That is, $K : \Sigma \to \mathbb{Z}$ such that $K(b) = 1$ for all $b \in \Sigma$.
Notice that if we add $K$ to a histogram, the difference is not affected. Therefore, we can consider the histograms in the quotient space
\[
\frac{\Sigma \to \mathbb{Z}}{\langle K \rangle}.
\]

\subsection{Incremented Cycle Register Rule}
\begin{definition}[Incremented Cycle Register Rule]
    Given a positive integer $k$, we define $\icr_k: \Sigma^n \rightarrow \Sigma^n$ as given by
    \[\icr_k(s) = \at{s}{1}\at{s}{2}\dots \at{s}{n-1} (\at{s}{0}+k).
    \] 
    Then, $\icr_k(s)$ is the string consisting of a cyclic rotation of the original $s$, but where the last symbol is incremented by $k$. Unless otherwise specified, $\icr = \icr_1$.
\end{definition}

Consider the subgraph of $B_n$ given by 
$(\Sigma^n, \{(s, ICR(s)):s\in \Sigma^n\})$, which we call the 
ICR-subgraph of $B_n$.  Since $\icr$ is a bijective function, the ICR-subgraph consists of a collection of node-disjoint cycles.

Consider a path $s_0, s_1, \dots, s_k$ in the $ICR$-subgraph. Can we reconstruct the histogram of first symbols of these strings knowing only $s_0$ and $s_k$?

\begin{lemma} \label{lemma:P}
    Let $s \in \Sigma^n$ and let $b = \at{s}{0}$ be the first symbol of $s$, then
    \[\Pa(\Hi(s)) - \Pa(\Hi(\icr(s))) \equiv e_b \mod K,\]
where $e_b$ is the indicator function at $b$.
\end{lemma}

\begin{proof}
By definition, $\Hi(s) - \Hi(\icr(s)) = e_b - e_{b+1}$.
If $b+1 \neq 0$, then $\Pa(e_b - e_{b+1}) = e_b$ and by linearity of $P$ we are done. If $b+1 = 0$, then $\Pa(e_b - e_{b+1}) = e_b - K$.
\end{proof}

\begin{lemma}
Let $s_0, \dots, s_k$ be a sequence of strings such that $s_{i+1} = \icr(s_i)$ for all $i$. Then the following equality holds
\[\Pa(\Hi(s_0)) - \Pa(\Hi(s_k)) \equiv \sum_{i=0}^{k-1} e_{\at{s_i}{0}} \mod K.
\]
\end{lemma}

\begin{proof}
    Apply Lemma~\ref{lemma:P} and use the telescopic sum.
\end{proof}

\subsection{Valid Subgraph}
Since the ICR-subgraph of $B_n$ is not enough to build a Hamiltonian cycle, 
we introduce another degree of freedom.
We consider  $d:\Sigma^n\to\Sigma$ to be a depth assignment, and for each $s\in\Sigma^n$, $d_s$ is its depth.

\begin{definition}[Valid Subgraph]
Given a depth assignment $d : \Sigma^n \to \Sigma$, an arc $(s, t)$ in the de Bruijn graph $B_n$ is called \emph{valid} if, for  $b = \at{s}{0}$ and $c = \at{t}{n-1}$,
\begin{itemize}
\item $b+1 = c$ and $d_s = d_t$, or
\item $b+1 = d_t$ and $c = d_s$
\end{itemize}
The set of valid arcs in $B_n$ form the valid subgraph.
\end{definition}

In Figure~\ref{fig:valid} there is an example of a valid subgraph.

\begin{figure}[htbp]
  \centering
  \subfloat[Valid subgraph. The dashed lines are the non-valid edges and the lines with arrows are the valid edges that don't belong to any $\icr$ cycle. Solid colored lines are the $\icr$ cycles.]{
    \includegraphics[scale=0.40]{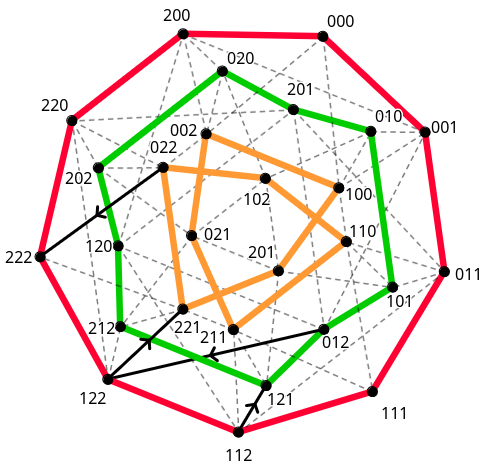}
   }
   \hspace{10pt}
     \subfloat[Hamiltonian cycle subgraph of the valid subgraph. Solid lines form the Hamiltonian cycle.]{

    \includegraphics[scale=0.40]{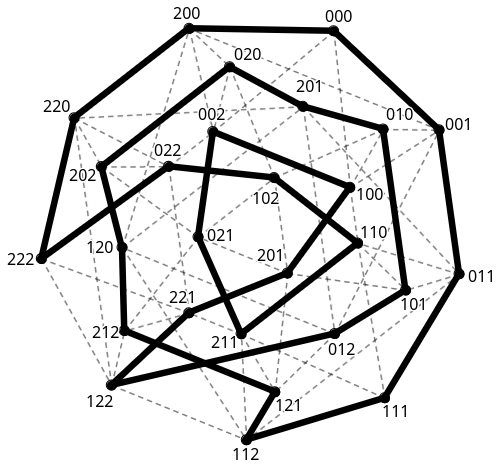}
    }
  \caption{An example of a valid subgraph for $|\Sigma| = 3$ and $n = 3$ where the $\icr$ cycle of $000$ is assigned depth $1$ and the other two $\icr$ cycles are assigned depth $2$.}
  \label{fig:valid}
\end{figure}

\begin{lemma}
    \label{validarcsprop}
    Given a depth assignment $d : \Sigma^n \to \Sigma$, and $s,t \in \Sigma^n$ such that $(s,t)$ is a valid arc in the de Bruijn graph then:
    \[e_{\at{s}{0}} \equiv \Pa(\Hi(s) + e_{d_s}) - \Pa(\Hi(t) + e_{d_t}) \mod K.\]
\end{lemma}
\begin{proof}
     Let $b = \at{s}{0}$ and $c = \at{t}{n-1}$. 
 Observe  that for any valid arc:
     $$e_b - e_c + e_{d_s} - e_{d_t} = e_b - e_{b+1}.$$
     Now,
    \begin{align*}
        \Pa(\Hi(s) + e_{d_s}) - \Pa(\Hi(t) + e_{d_t})  &\equiv \Pa\left(\Hi(s) - \Hi(t) + e_{d_s} - e_{d_t}\right) &\mod K\\
        &\equiv \Pa(e_b - e_c + e_{d_s} - e_{d_t}) &\mod K \\
        &\equiv \Pa(e_b - e_{b+1}) &\mod K \\
        &\equiv e_b \equiv e_{s[0]} &\mod K.\ \ \qedhere
\end{align*}
\end{proof}


%

If the depth assignment is constant in the cycles produced by the ICR rule, then the ICR-subgraph is a subgraph of the valid subgraph. If we have any path in the valid subgraph, then a telescoping equality holds.

\begin{lemma}
    \label{thm:bvalid}
    Let $d : \Sigma^n \to \Sigma$ be a depth assignment and let $s_0, \dots, s_k$ be a path in the valid subgraph, then 
\[
\sum_{i=0}^{k-1} e_{\at{s_i}{0}} \equiv \Pa\left(\Hi(s_{0}) + e_{d_{s_{0}}}\right) - \Pa\left(\Hi(s_{k}) + e_{d_{s_{k}}}\right) \mod K.
\]
\end{lemma}

The histogram of the sequence of first symbols in a path 
in the valid subgraph on $B_n$ 
depends only on the initial and final string. 
Observe that in the formula for Lemma~\ref{thm:bvalid}, the histogram ignores the first symbol of the last string (the summation only goes up to $k-1$). This is not a problem for our application: Suppose that, for every pair of strings $s, t \in \Sigma^n$ we manage to prove a bound
\[    \Di\left(\Pa\left(\Hi(s) + e_{d_{s}}\right) - \Pa\left(\Hi(t) + e_{d_{t}}\right)\right) \leq F, \quad \quad F \in \mathbb{N}
\]
Then, every de Bruijn sequence that arises from a Hamiltonian cycle $s_0,\dots,s_k$ in the valid subgraph, has discrepancy-bound $F$. Indeed, if we consider the substring $\at{s_i}{0},\at{s_{i+1}}{0}\dots,\at{s_j}{0}$, then we can bound
\[\Di(\Hi(\at{s_i}{0},\at{s_{i+1}}{0}\dots,\at{s_j}{0}))\]
by applying 
Lemma~\ref{thm:bvalid} on the path $s_i,s_{i+1},\dots,s_j,s_{j+1}$ with an extra past-the-end element.

\subsection{Difference Bounds for Valid Subgraph Paths}

We start by ignoring the depths.

\begin{lemma}[Difference of partial sums of histograms]
    \label{thm:dpartial}
    If $s, t \in \Sigma^n$ then
    \[\Di(\Pa(\Hi(s) - \Hi(t))) \leq n.\]
\end{lemma}

\begin{proof}
    The histogram $\Hi(s)$ can be written as a sum over all symbols of $s$,
    \[
    \Hi(s) = \sum_{i=0}^{n-1} e_{\at{s}{i}}.
    \]
    Similarly,
    \[\Hi(t) = \sum_{i=0}^{n-1} e_{\at{t}{i}}.
    \]
    Therefore,
    \[\Pa(\Hi(s) - \Hi(t)) = \sum_{i=0}^{n-1} \Pa(e_{\at{s}{i}} - e_{\at{t}{i}}).
    \]
    Notice that $\Pa(e_{\at{s}{i}} - e_{\at{t}{i}})$ has a difference of at most $1$ 
    because its values are either all $0$s and $1$s, or all $0$s and $-1$s. Therefore, due to subadditivity we have the desired bound.
\end{proof}

\begin{lemma}[Difference of partial sums of histograms with shared symbol]
    \label{thm:dpartial2}
    If $s, t \in \Sigma^n$ and there exists a symbol that appears in both strings, then
    \[\Di(\Pa(\Hi(s) - \Hi(t))) \leq n-1.\]
\end{lemma}
\begin{proof}

In the proof of Lemma~\ref{thm:dpartial} we obtained that when $s,t\in \Sigma^n$, 
    \[\Pa(\Hi(s) - \Hi(t)) = \sum_{i=0}^{n-1} \Pa(e_{\at{s}{i}} - e_{\at{t}{i}}).\]
    Since the histograms do not depend on the order of the symbols, we can assume without loss of generality that $\at{s}{0} = \at{t}{0}$, and therefore the first summand of the sum vanishes, so the remaining $n-1$ can contribute at most $1$ difference each.
\end{proof}

Now we consider difference of partial sums  but including depths.
\begin{lemma}
    If $s, t \in \Sigma^n$ are strings and $d : \Sigma^n \to \Sigma$ is a depth assignment, then
    \[\Di\left(\Pa\left(\Hi(s) + e_{d_{s}}\right) - \Pa\left(\Hi(t) + e_{d_{t}}\right)\right) \leq n+1.\]
\end{lemma}

\begin{proof}
    The formula follows from Lemma~\ref{thm:dpartial} using subadditivity of $D$, linearity of $P$ and the fact that $\Di(\Pa(e_{d_s} - e_{d_t})) \leq 1$.
\end{proof}

\begin{lemma}
    If $s, t \in \Sigma^n$ are strings and $d : \Sigma^n \to \Sigma$ is a depth assignment, such that
    \begin{itemize}
    \item $d_s = d_t$, or
    \item $s$ and $t$ share at least one symbol
    \end{itemize}
    Then
    \[\Di\left(\Pa\left(\Hi(s) + e_{d_{s}}\right) - \Pa\left(\Hi(t) + e_{d_{t}}\right)\right) \leq n.\]
\end{lemma}

\begin{proof}
In the case where $d_s = d_t$, the term $\Pa(e_{d_s}) - \Pa(e_{d_t})$ vanishes and contributes $0$ to the total difference, so the difference-bound of $n$ for the non-depth terms applies.

In the case where $s$ and $t$ share a symbol, the $n-1$ bound on the non-depth terms applies, with at most an extra unit of difference from the $\Pa(e_{d_s}) - \Pa(e_{d_t})$ term.
\end{proof}

The  results above lead us to the following two theorems.

\begin{theorem}
    \label{thm:bound1}
    If there exists a depth assignment $d : \Sigma^n \to \Sigma$ such that the valid subgraph of $B_n$  has a Hamiltonian cycle, then 
    such a  
    cycle corresponds to a de Bruijn sequence of order~$n$ with discrepancy at most~$n+1$.
\end{theorem}

\begin{theorem}
    \label{thm:bound2}
    If there exists a depth assignment $d : \Sigma^n \to \Sigma$ such that
    \begin{itemize}
        \item every pair of strings with different depth share a symbol, and
        \item the valid subgraph of $B_n$  has a Hamiltonian cycle,
    \end{itemize}
    then that cycle corresponds to a de Bruijn sequence of order $n$ with discrepancy at most $n$.
\end{theorem}
For $|\Sigma| = 2$, the first condition in
Theorem~\ref{thm:bound2} is satisfied if and only if the string with all zeros and the string with all ones are assigned the same depth.

\section{Our Hamiltonian Cycle}

Since $\icr_1 : \Sigma^n \to \Sigma^n$ is a bijective function, 
it partitions the space $\Sigma^n$ into cycles. Let $\mathbf{N}$ be the set of cycles given by the $\icr_1$ rule. We use this set to build a de Bruijn sequence by joining these cycles to obtain a Hamiltonian cycle.

\begin{definition}[Orbit of a node]
     We define $\orbit: \Sigma^n \rightarrow \mathcal{P}(\Sigma^n)$ as
    \[\orbit(s) = \{s, \icr(s), \icr^2(s), \dots \}.
    \]
\end{definition}
 
Alternatively, it can be defined as the equivalence class of the relation defined by: 
\[
 R(s,t) \iff t = \icr^k(s)  \text{ for some } k\in \mathbb{Z} .
\]

\begin{definition}[Cycle Graph]
     We define a directed graph $\mathbf{G}$ with node-set $\mathbf{N}$ and an arc $(\mathcal{U}, \mathcal{V})$ between cycles $\mathcal{U}, \mathcal{V} \in \mathbf{N}$ if, and only if, there exists an $s \in \mathcal{U}$ such that $\icr_0(s) \in \mathcal{V}$.
\end{definition}

Let us assume that there exists a directed tree $\mathbf{T}$ rooted at $\mathcal{R} \in \mathbf{N}$ that spans $\mathbf{G}$ (the tree is directed from the root to the leaves). The next section proves the existence of $\mathbf{T}$. The tree $\mathbf{T}$ defines a parent relationship $\parent : \mathbf{N} \setminus \{\mathcal{R}\} \to \mathbf{N}$.

\begin{definition}[Representative]
    Let $\mathcal{U} \in \mathbf{N} \setminus \{\mathcal{R}\}$. The representative $s$ of $\mathcal{U}$ is any (fixed) node $s \in \mathcal{U}$ such that $\icr_0^{-1}(s) \in \parent(\mathcal{U})$ and $\at{s}{n-1} \equiv d \mod |\Sigma|$, where $d$ is the depth of $\mathcal{U}$ in the tree.
    We write 
    $\rep(\mathcal{U}) = s$,
    and also
    $\rep(u) = s$
    for any $u \in \mathcal{U}$. The set of representatives of   all cycles is denoted with $\reps$.
\end{definition}

Observe  that for all $s \in \reps$, $\icr_0^{-1}(s) \in \parent(\orbit(s))$.
Let us define the following rule:
\[
    \nxt(s) = \begin{cases}
        \icr_0(s) & \text{ if } \icr_0(s) \in \reps\\
        \icr_2(s) & \text{ if } \icr_1(s) \in \reps\\
        \icr_1(s) & \text{ otherwise.}
    \end{cases}
\]

\begin{lemma}\label{lemma:Rfunction}
$\nxt$ is a function.
\end{lemma}
\begin{proof}
We have to show  that  there is no
$s \in \Sigma^n$ such that $\icr_0(s), \icr_1(s) \in \reps$. We only need to prove the case $|\Sigma|>2$, 
because for $|\Sigma| = 2$, $\icr_0$ and $\icr_2$ coincide.

Suppose $\icr_0(s), \icr_1(s) \in \reps$. Then, by definition, $\orbit(s) = \orbit(\icr_1(s))$. And, since $\icr_0(s) \in \reps$, necessarily $\orbit(s)$ is the parent of $\orbit(\icr_0(s))$. This implies that the depth of $\orbit(\icr_0(s))$ is one more than the depth of $\orbit(\icr_1(s))$.
On the other hand, the last symbol of $\icr_1(s)$ is one more than the last symbol of $\icr_0(s)$, and by definition they match the depths of their respective orbits modulo $|\Sigma|$. Since $|\Sigma| > 2$, this is a contradiction.
\end{proof}

Let us prove the following lemmas about this rule.
\begin{lemma}[Bijectivity]
The function $\nxt$ is bijective.
\end{lemma}

\begin{proof}
Suppose $\nxt(s) = \nxt(t)$, then necessarily the maximal proper suffixes of $s$ and $t$ coincide. It is easy to see that if the rule applied to $s$ and $t$ is the same, then $s = t$, because all $\icr_k$ are injective. Then, we have three cases without symmetries:
\begin{itemize}
    \item $\nxt(s) = \icr_0(s) = \icr_1(t) = \nxt(t)$: 
    Since $R(s) = \icr_0(s)$, then $\nxt(t) = \nxt(s) \in \reps$. But this is impossible because $R(t) = \icr_1(t)$.
    
    
     \item $\nxt(s) = \icr_0(s) = \icr_2(t) = \nxt(t)$: we can see that $\icr_1(\icr_0^{-1}(\icr_1(t))) = \icr_2(t)$ for all $t$. Then, $\icr_0^{- 1}(\icr_1(t)) = \icr_1^{-1}(\icr_0(s))$, let's call this string~$u$. Now, we know that $\icr_0(s), \icr_1(t) \in \reps$ so $\icr_0(u), \icr_1(u) \in \reps$. If $|\Sigma| > 2$ this is impossible, 
     as shown in Lemma~\ref{lemma:Rfunction}.
 
     \item $\nxt(s) = \icr_1(s) = \icr_2(t) = \nxt(t)$: this means that $\icr_1(t) \in \reps$, but this is impossible 
     because $\icr_0(s) = \icr_1(t)$ and then $\icr_0(s) \in \reps$.
\end{itemize}
\end{proof}


\begin{lemma}[Transitivity]
    The bijection $\nxt$ is transitive. That is, for any $s, t \in \Sigma^n$ there exists an integer $k$ such that $\nxt^k(s) = t$.
\end{lemma}

\begin{proof}
    Since $\nxt$ is a bijection, it partitions $\Sigma^n$ into disjoint $\nxt$-cycles. Suppose there is an $\icr$-cycle $\mathcal{U} \in \mathbf{N}$ such that $\mathcal{U}$ is not completely contained in any $\nxt$-cycle. That means $\mathcal{U}$ contains two nodes belonging to different $\nxt$-cycles. In particular, there are two strings $s,t \in \mathcal{U}$ such that $s$ and $\icr_1(s)$ belong to different $\nxt$-cycles and, also, $t$ and $\icr_1(t)$  belong to different $\nxt$-cycles.  
    Thus, without loss of generality, we can assume that $\icr_1(s) \notin \reps$ (since $\mathcal{U} \cap \reps$ has size at most one), and also that $\mathcal{U}$ is the deepest node in the tree that intersects two different $\nxt$-cycles.

    Clearly, $\nxt(s) \neq \icr_1(s)$, otherwise $s$ and $\icr_1(s)$ would belong to the same $\nxt$-cycle. Since $\icr_1(s) \notin \reps$, the remaining possibility (from the definition of $\nxt$) is that $\icr_0(s) \in \reps$. From the definition of $\reps$, we get that the orbit of $\icr_0(s)$ is a child of the orbit of $s$.
    From the definition of $\nxt$, we get that
    \[\nxt(\icr_1^{-1}(\icr_0(s))) = \icr_2(\icr_1^{-1}(\icr_0(s))) = \icr_1(s),\]
    and therefore $\icr_1^{-1}(\icr_0(s))$ and $\icr_1(s)$ belong to the same $\nxt$-cycle. But $\icr_0(s)$ belongs to the same 
    $\nxt$-cycle as $s$, so the orbit of $\icr_0(s)$ \emph{also} contains two different $\nxt$-cycles, which contradicts the assumption that $\mathcal{U}$ is the deepest.

    Now assume that there are two adjacent orbits $\mathcal{U}, \mathcal{V} \in \mathbf{N}$ in the tree that are contained in different $\nxt$-cycles. Without loss of generality, let us assume that $\mathcal{V}$ is the child of $\mathcal{U}$ and let $s = \rep(\mathcal{V})$.
    Since $s \in \reps$, we know that $\nxt(\icr_0^{-1}(s)) = s$ and that $\icr_0^{-1}(s)$ is in the parent orbit of $s$ (namely, $\mathcal{U}$). This is a contradiction, because 
    we assumed that $\mathcal{U}$ and $\mathcal{V}$ do not belong to the same $\nxt$-cycle.
\end{proof}

As $\nxt$ is a bijective and transitive function, and it also satisfies the property that the arc \mbox{$(s, \nxt(s))$} is in the de Bruijn graph, we can construct a Hamiltonian cycle by taking an arbitrary node as the start, let us call it $s$, and then repeatedly applying the function $\nxt$ until we arrive to $s$ again.

\subsection{Hamiltonian cycle in the Valid Subgraph}

Now that we have proved that $\nxt$ generates a Hamiltonian cycle 
in the de Bruijn graph, let us prove that 
all of its arcs are 
in the valid subgraph. To construct the valid graph we 
need to define  a depth assignment.

\begin{definition}[Depth Assignment]
    For a fixed tree of ICR-cycles, we define
    \[d : \Sigma^n \to \Sigma \quad \quad d(s) = p+1,\]
    where $p$ is the depth of the orbit of $s$ in the tree, modulo $|\Sigma|$.
\end{definition}

The resulting valid subgraph for the case where the parent of every cycle $ \mathcal{U} \in \mathbf{N} \setminus \{\mathcal{R}\}$ is $\mathcal{R}$ is given in Figure~\ref{fig:valid}, and the unique Hamiltonian path in the subgraph is shown.

\begin{lemma}
    The arc $(s, \nxt(s))$ is valid.
\end{lemma}

\begin{proof}
    Let $b$ be the first symbol of $s$ and $c$ be the last symbol of $\nxt(s)$. We need to consider the three cases in the definition of $\nxt(s)$:
    \begin{itemize}
        \item Case $\icr_0(s) \in \reps$: In this case $\nxt(s) = \icr_0(s)$ and therefore $c = b$ and $c = d_{\nxt(s)}-1$. Also, by definition of $\reps$, the orbit of $\nxt(s)$ is a child of the orbit of $s$, so we have $d_s + 1 = d_{\nxt(s)}$. Putting these together we get that $b+1 = d_{\nxt(s)}$ and $c = d_s$, so the arc is valid.
        \item Case $\icr_1(s) \in \reps$: In this case $\nxt(s) = \icr_2(s)$ and therefore $c = b+2$, and 
        given that $\icr_1(s) \in \reps$, $b+1 = d_{s}-1$. Since $\icr_2(s) = \icr_1(\icr_0^{-1}(\icr_1(s)))$, the orbit of $\icr_2(s)$ is the parent of the orbit of $s$, and thus $d_{\nxt(s)} = d_s - 1$. Putting these together we get $b+1 = d_{\nxt(s)}$ and $c = d_s$, so the arc is valid.
        \item Remaining case: We have $\nxt(s) = \icr_1(s)$, so $c = b+1$ and $d_{s} = d_{\nxt(s)}$, so the arc is once again valid.\qedhere
    \end{itemize}
\end{proof}

\subsection{Discrepancy Bound for the General Case}

Due to Theorem~\ref{thm:bound1}, we have that the de Bruijn sequence associated with $\nxt$ has discrepancy at most $n+1$.

\subsection{Discrepancy Bound for the Binary Case}

For the binary case, we  use Theorem~\ref{thm:bound2}: since the strings $0^n$ and $1^n$ 
belong to the same orbit, they are assigned the same depth value and therefore the hypotheses of 
the theorem hold and we have that the de Bruijn sequence associated with $\nxt$ has discrepancy at most $n$.

\subsection{Explicit Spanning Tree}

To complete the proof, we have
to show a directed spanning tree $\mathbf{T}$  of $\mathbf{G}$ and choose the representatives for each orbit. To do this we define the following:


\begin{definition}[Difference Array]
    Given $s \in \Sigma^n$, we define its difference array $\dif(s) \in \Sigma^n$, 
    \[\at{\dif(s)}{i} = \begin{cases}
        \at{s}{i-1} - \at{s}{i} & \text{ if } 0 < i < n\\
        \at{s}{n-1} - \at{s}{0} - 1 & \text{ if } i = 0.
    \end{cases}\]
\end{definition}

The extra $-1$ in the $i=0$ case ensures the following property:

\begin{lemma}
    \label{thm:diff}
    Let $s \in \Sigma^n$, then $\dif(\icr_1(s)) = \icr_0(\dif(s))$.
\end{lemma}


\begin{proof}
Consider the bidirectional infinite sequence
\[c_i = \at{\icr^i(s)}{0}, \quad\quad i \in \mathbb{Z}.
\]
By the definition of $\icr$, we have that $c_{i+n} = c_i + 1$. Now consider the following associated (infinite) sequence
\[d_{i} = c_{i-1} - c_i, \quad\quad i \in \mathbb{Z}.
\]
Since $c_{i+n} = c_i + 1$, we have that $d$ is cyclic modulo $n$. Also, due to the definition of $\dif$ we have that
\[\dif(s) = d_0 d_1 \dots d_{n-1} \quad\text{and}\quad \dif(\icr(s)) = d_1 d_2 \dots d_n.
\]
Since $d_n = d_0$, the desired result follows.
\end{proof}


Lemma~\ref{thm:diff} implies that for each $\mathcal{U} \in \mathbf{N}$, its elements have the same difference array modulo rotations. By convention, we say that $\dif(\mathcal{U})$ is the lexicographically minimal rotation of $\dif(s)$ for any $s \in \mathcal{U}$. It is easy to see that the converse is also true: for any $s,t \in \Sigma^n$, if $\dif(s)$ and $\dif(t)$ are equal modulo rotations, then $\orbit(s) = \orbit(t)$. Indeed, consider the string $u = \icr_1^k(s)$. We can vary $k$ such that $\dif(u) = \dif(t)$, and then the strings $u$ and $t$ only differ by an added constant, so they must belong to the same orbit.

Also, 
from the definition of $\dif$, we have the following identity
\[\sum_{i=0}^{n-1} \at{\dif(s)}{i} = -1,\]
which follows from the fact that every symbol of $s$ cancels out in the sum, except the constant factor $-1$ in the first position of $\dif(s)$. Again, the converse is also true. If we have any string $t \in \Sigma^n$ such that
    \[\sum_{i=0}^{n-1} \at{t}{i} = -1,\]
then there exists a $s \in \Sigma^n$ that satisfies $\dif(s) = t$, which is unique modulo added constants.
To see this, we can choose $\at{s}{0}$ arbitrarily, then once $\at{s}{i}$ has been determined, the value of $\at{s}{i+1}$ follows directly from the equation
    \[\at{\dif(s)}{i+1} = \at{t}{i+1}.
    \]
The only equation that 
remains to be satisfied is $\at{\dif(s)}{0} = \at{t}{0}$. However,
$\dif(s)$ and $t$ coincide in $n-1$ places and they have the same sum, $-1$, so they must be identical.



Equipped with the concept of difference array, we can define the tree $\mathbf{T}$ 
as follows.




\begin{definition}[Explicit Tree $\mathbf{T}$]
    
    The root of the tree is the orbit of $0^n$ and we call it $\mathcal{R}$.
    For each orbit $\mathcal{U} \in \mathbf{N} \setminus \mathcal{R}$, we define its parent $\mathcal{V} \in \mathbf{N}$ as follows:

    let $i$ be the first non-negative integer such that $\at{\dif(\mathcal{U})}{i} \neq 0$. Consider the array $A$ obtained from decrementing the $i$-th symbol of $\dif(\mathcal{U})$ and incrementing the $(i+1)$-th symbol. $\mathcal{V}$ is the only orbit such that $\dif(\mathcal{V})$ equals $A$ modulo rotations.
    \label{thm:extree}
\end{definition}

\begin{figure}
\centering
    \includegraphics[scale=0.06]{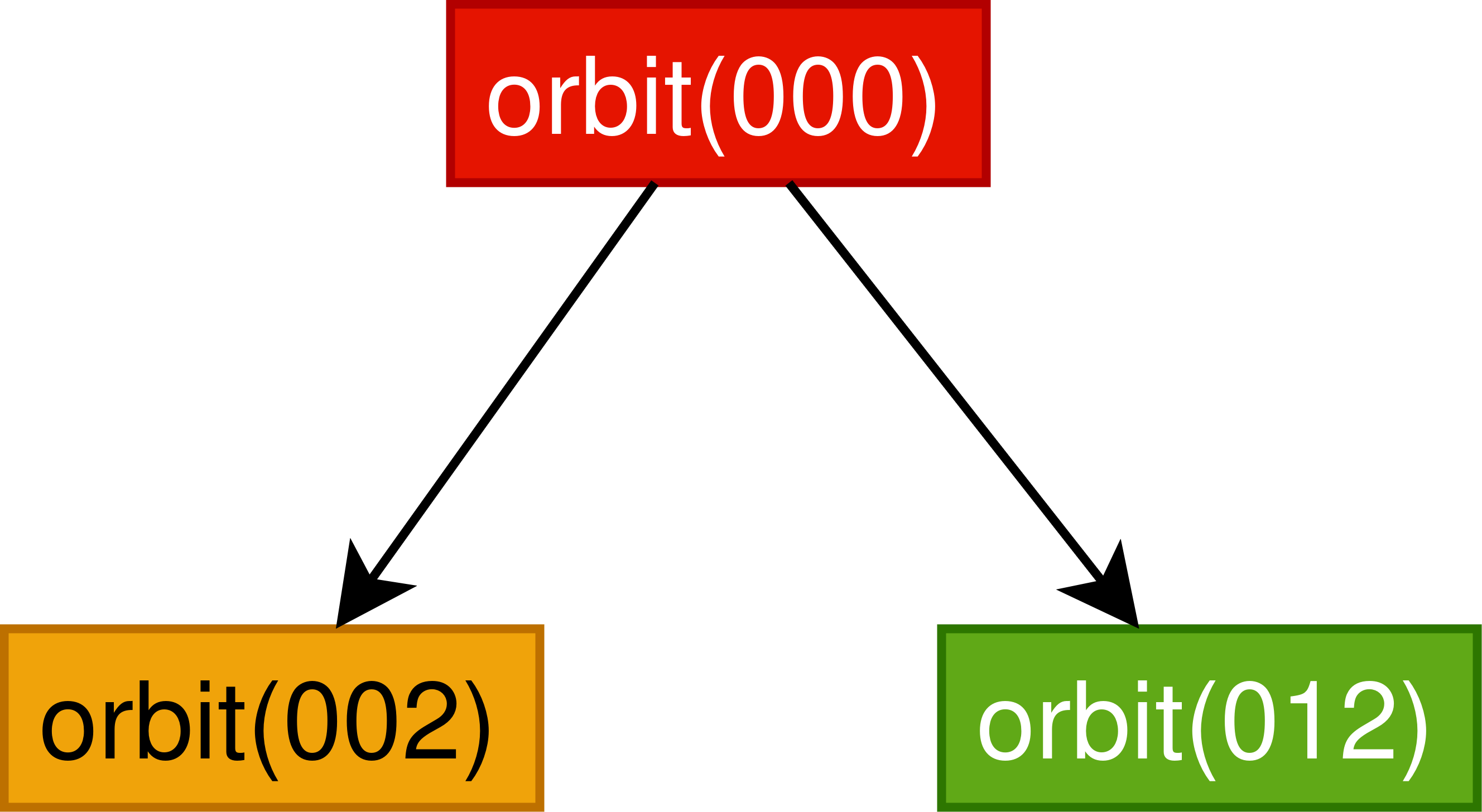}
    \caption{Explicit Tree for the case $n=3, |\Sigma| = 3$.}
    \label{fig:tree}
\end{figure}

Figure~\ref{fig:tree} gives an example of an Explicit Tree $\mathbf{T}$. 

\begin{lemma}
The Explicit Tree $\mathbf{T}$ is a spanning tree of $\mathbf{G}$.
\end{lemma}
\begin{proof}

The existence of~$i$ follows from the fact that the sum of elements of the difference array is $-1$, so it cannot consist solely of zeros. In fact, $i \leq n-2$;
otherwise in $\dif(\mathcal{U})$ there would be $n-1$ zeros and, necessarily, the remaining symbol would be $-1$, which is the difference array of $\mathcal{R}$.

To prove that the parent relationship has no cycles, note that $A$ is lexicographically smaller than $\dif(\mathcal{U})$, so we have that
\[\dif(\mathcal{V}) \leq A < \dif(\mathcal{U})
\]
and this ensures that there are no cycles.



\medskip
Let us see that for every node $\mathcal{U} \in \mathbf{N} \setminus \mathcal{R}$, the pair $(\mathcal{V}, \mathcal{U})$ is an arc of $\mathbf{G}$ when $\mathcal{V}$ is the parent of $\mathcal{U}$ as defined in the definition of the Explicit Tree.
It suffices to find an $s \in \mathcal{U}$ such that $\icr_0^{-1}(s) \in \mathcal{V}$, which can also serve as a representative of $\mathcal{U}$.
Let $i$ be the first non-negative integer such that $\at{\dif(\mathcal{U})}{i} \neq 0$, and consider the array $k = \icr_0^{i+1}(\dif(\mathcal{U}))$. That is, the array $\dif(\mathcal{U})$ rotated such that the first non-zero element moves to the last position. Let $s \in \Sigma^n$ be a string with $\dif(s) = k$. We can choose the added constant of $s$ as necessary to satisfy the last-symbol constraint in the definition of representative.

We want to show that $\icr_0^{-1}(s) \in \mathcal{V}$. To do this, let us prove that
$\dif(\icr_1(\icr_0^{-1}(s)))$ is a rotation of $A$, this will imply that $\icr_1(\icr_0^{-1}(s)) \in \mathcal{V}$ and since it is in the same orbit as $\icr_0^{-1}(s)$, it would mean that the latter is also in $\mathcal{V}$.

Observe that $\icr_1(\icr_0^{-1}(s))$ is the same as $s$ but with the last symbol increased by one. The effect that increasing the last symbol of $s$ has on $\dif(s)$ is that of decreasing the last symbol and increasing the first symbol.
Since $\dif(s) = \icr_0^{i+1}(\dif(\mathcal{U}))$, we have that $\dif(\icr_1(\icr_0^{-1}(s)))$ is the same as $\icr_0^{i+1}(\dif(\mathcal{U}))$ but the last symbol decreased and the first one increased, which is precisely the same as $\icr_0^{i+1}(A)$. This concludes the proof.
\end{proof}

\section{Our Algorithm}

\begin{algorithm}[ht]
    \DontPrintSemicolon
    \caption{CorrectDifferenceArray \label{alg:isrep}}
    
    \KwData{$s \in \Sigma^n$}
    \KwResult{\\True if $s$ has the correct difference array to be a representative, False otherwise}
    
    $d \leftarrow \dif(s)$

    \If{$\at{d}{n-1} = 0$}{
        \Return{False}
    }

    $i \leftarrow n-1$
    
    \While{$i > 0$ and $\at{d}{i-1} = 0$}{
        $i \leftarrow i-1$
    }
    
    \If{$i = 0$}{
        \Return{False}
    }
    
    $d_2 \leftarrow \icr_0^{i}(d)$

    \Return{$d_2$ is its lexicographically minimal rotation}
\end{algorithm}

\begin{algorithm}[ht]
    \DontPrintSemicolon
    \caption{Transition \label{alg:transition}}
    
    \KwData{$s \in \Sigma^n$, $\text{depth} \in \Sigma$}
    \KwResult{A Pair $(s', \text{depth}')$ that corresponds to the next node in the path}
    
    \If{$\text{CorrectDifferenceArray}(\icr_0(s))$ and $\at{\icr_0(s)}{n-1} = \text{depth}+1$}{
        \Return{$\icr_0(s), \text{depth}+1$}
    }
    \If{$\text{CorrectDifferenceArray}(\icr_1(s))$ and $\at{\icr_1(s)}{n-1} = \text{depth}$}{
        \Return{$\icr_2(s), \text{depth}-1$}
    }
    \Return{$\icr_1(s), \text{depth}$}
\end{algorithm}

The algorithm constructs a de Bruijn sequence of order $n$ in the alphabet $\Sigma$.

Recall that we defined the transition rule $\nxt$ as follows,

\[
    \nxt(s) = \begin{cases}
        \icr_0(s) & \text{ if } \icr_0(s) \in \reps\\
        \icr_2(s) & \text{ if } \icr_1(s) \in \reps\\
        \icr_1(s) & \text{ otherwise.}
    \end{cases}
\]

For the algorithm to run in $O(n)$ space, 
we cannot maintain the tree in memory. Instead, we use the definition of the Explicit Tree ${\mathbf T}$ and we give an efficient procedure to compute membership to the set $\reps$.
This has two parts:
\begin{itemize}
    \item The difference array of the string must be of the correct form. That is, it should be the lexicographically minimal rotation, shifted so that the first non-zero element is at the end. To be a representative, $s$ also must not be in the root orbit, which we can also check with the difference array.
    \item The last symbol of the string must be equal to the depth of the orbit modulo $|\Sigma|$.
\end{itemize}

The first part is implemented in Algorithm~\ref{alg:isrep}. The second part is implemented together with the transition function in Algorithm~\ref{alg:transition}. Note that both functions have linear complexity in both time and memory. To determine whether $d_2$ is the lexicographically minimal rotation we can use Booth's algorithm~\cite{Booth}, for example.

\subsection{Example Sequences}
In Table~\ref{table:tex} and Figure~\ref{fig:ex} (left) we display the output of our algorithm.

\begin{table}[h!]
\centering
\begin{tabular}{|c|c|l|}
\hline
$|\Sigma|$ & \textbf{n} & \textbf{Resulting Sequence}\\
\hline
$2$ & $2$ & \texttt{1100}\\
\hline
$2$ & $3$ & \texttt{11101000}\\
\hline
$2$ & $4$ & \texttt{1111001011010000}\\
\hline
$2$ & $5$ & \texttt{11111000101011001001101110100000}\\
\hline
$2$ & $6$ & \texttt{1111110001001100111011000010110101001010111001000110111101000000}\\
\hline
$3$ & $2$ & \texttt{112102200}\\
\hline
$3$ & $3$ & \texttt{111212020101221002110222000}\\
\hline
$4$ & $2$ & \texttt{1121320310223300}\\
\hline
$4$ & $3$ & \texttt{1112123230201312023130301012213320021132203310321003110222333000}\\
\hline
\end{tabular}
\caption{Example sequences produced by our algorithm.}
\label{table:tex}
\end{table}

\section{Analysis}
\subsection{Our Algorithm Behavior for an Arbitrary Base}
For base $|\Sigma|=2$, we know our algorithm produces de Bruijn sequences of discrepancy exactly~$n$, because that is both the upper bound and the trivial lower bound. For the case $n=1$, we also know our algorithm produces a sequence of discrepancy exactly $n$, because all de Bruijn sequences for $n=1$ have discrepancy $n$. Therefore, in both of these cases our algorithm is optimal.


However, for the case $|\Sigma| > 2$ and $n>1$, it is still unknown whether the minimum attainable discrepancy is $n$ or $n+1$.
Our algorithm produces a de Bruijn sequence with discrepancy exactly $n+1$ (there is a  proof, not included in this paper, that shows our method does not achieve discrepancy equal to~$n$).




\subsection{Conjecture on the Minimum Attainable Discrepancy}

\begin{table}[h!]
\centering
\begin{tabular}{|c|*{8}{c|}}
\hline
\diagbox{\textbf{$|\Sigma|$}}{\textbf{n}} & \textbf{1} & \textbf{2} & \textbf{3} & \textbf{4} & \textbf{5 } & \textbf{6} & \textbf{7}\\
\hline
\textbf{2} & n & n & n & n & n & n & n \\
\hline
\textbf{3} & n & n+1 & n & n &  &  &  \\
\hline
\textbf{4} & n & n+1 & n &  &  &  &  \\
\hline
\textbf{5} & n & n &  &  &  &  &  \\
\hline
\textbf{6} & n & n &  &  &  &  &  \\
\hline
\textbf{7} & n & n &  &  &  &  &  \\
\hline
\textbf{8} & n & n &  &  &  &  &  \\
\hline
\end{tabular}
\caption{Each cell has the minimum discrepancy attainable for a de Bruijn sequence with the corresponding parameters. The cases that could not be computed are blank.}
\label{table:order_base}
\end{table}

To compare the behaviour of the algorithm with the actual minimum discrepancy attainable, we  made an exhaustive search for the smallest discrepancy that can be obtained with certain parameters of $n$ and $|\Sigma|$, to decide whether our algorithm was optimal or not. The results we obtained can be seen in Table~\ref{table:order_base}.
The experimentation is heavily limited due to the doubly exponential nature of the search. 

Based on the experimental results and
 given that the de Bruijn graph gets more interconnected with increasing $n$ and $|\Sigma|$, one can  expect that achieving a low-discrepancy de Bruijn sequence becomes easier as these parameters increase.
We put forward the following conjecture.

\begin{conjecture}
 The minimal 
 discrepancy of a de Bruijn sequence of order $n$ and alphabet $\Sigma$ is $n$, with the exception of the cases   $ n = 2$ with $|\Sigma| = 3 $ and  $n = 2$ with $|\Sigma| = 4 $.
\end{conjecture}




\subsection{Comparison with Other de Bruijn Sequences}

\begin{table}
\centering
\begin{tabular}{|c|*{16}{c|}}
\hline
\textbf{n} & \textbf{Our Algorithm} & \textbf{Huang} & \textbf{Random} & \textbf{Weight-range}\\
\hline
$10$ & $10$ & $12$ & $50$ & $131$\\
\hline
$11$ & $11$ & $13$ & $71$ & $257$\\
\hline
$12$ & $12$ & $15$ & $101$ & $468$\\
\hline
$13$ & $13$ & $16$ & $143$ & $930$\\
\hline
$14$ & $14$ & $18$ & $203$ & $1723$\\
\hline
$15$ & $15$ & $19$ & $288$ & $3439$\\
\hline
$16$ & $16$ & $21$ & $407$ & $6443$\\
\hline
$17$ & $17$ & $22$ & $575$ & $12878$\\
\hline
$18$ & $18$ & $24$ & $815$ & $24319$\\
\hline
$19$ & $19$ & $25$ & $1157$ & $48629$\\
\hline
$20$ & $20$ & $27$ & $1634$ & $92388$\\
\hline
$21$ & $21$ & $28$ & $2311$ & $184766$\\
\hline
$22$ & $22$ & $30$ & $3264$ & $352727$\\
\hline
$23$ & $23$ & $31$ & $4565$ & $705443$\\
\hline
$24$ & $24$ & $33$ & $6252$ & $1352090$\\
\hline
$25$ & $25$ & $35$ & $9192$ & $2704168$\\
\hline
$26$ & $26$ & $36$ & $13074$ & $5200313$\\
\hline
$27$ & $27$ & $38$ & $17933$ & $10400613$\\
\hline
$28$ & $28$ & $40$ & $22672$ & $20058314$\\
\hline
$29$ & $29$ & $41$ & $34591$ & $40116614$\\
\hline
$30$ & $30$ & $43$ & $57357$ & $77558775$\\
\hline
\end{tabular}
\caption{Discrepancy attained by various algorithms for binary de Bruijn sequences. The data for \textbf{Huang}, \textbf{Random} and \textbf{Weight-range} is taken from~\cite{Sawada}.}
\label{table:art}
\end{table}

\begin{figure}[b!!]
\centering
\includegraphics[width=0.30\textwidth]{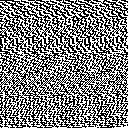}
\hspace{0.02\textwidth}
\includegraphics[width=0.30\textwidth]{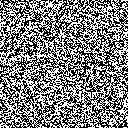}
\hspace{0.02\textwidth}
\includegraphics[width=0.30\textwidth]{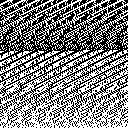}
\caption{Graphical representation for the binary  sequences of order $n=14$ produced by our algorithm (left), a de Bruijn sequence chosen uniformly at random (center) and the de Bruijn sequence with  asymptotically maximum discrepancy in~\cite{Sawada} (right). 
The symbols of the sequence are displayed in row-major order. Zero is white, one is black. }
\label{fig:ex}
\end{figure}

In Table~\ref{table:art} we compare the discrepancy obtained with our algorithm in a binary alphabet with that  obtained in previously-described algorithms in the literature:

\begin{itemize} \setlength\itemsep{0.1em}
    \item Huang's algorithm~\cite{Huang} which attained the previously smallest known discrepancy.
    \item A uniformly random de Bruijn sequence.
    
        \item The Weight-range construction~\cite{Sawada}, which is proven to have the asymptotically maximum discrepancy possible.
\end{itemize}

\clearpage
\section{Code}
We include a full \texttt{C++} implementation of the proposed algorithm.
\begin{multicols}{2}
\begin{lstlisting}[language=C++]
#include <vector>
#include <algorithm>
#include <iostream>
using namespace std;

// n is the order of the de Bruijn string
int n;

// base is the size of the alphabet
int base;

// Checks if s is its lexicographically minimal
// rotation in O(n). Modified Duval.
// Implementation from:
//   https://stackoverflow.com/a/73966629
bool min_lex(const vector<int> &s) {
    for(int i=0, j=1; j < 2*n; j++) {
        int a = s[i%n], b = s[j%n];
        if(b < a) return false;
        else if(a < b) i = 0;
        else i++;
    }
    return true;
}

// The difference array of s, defined as
//   res[i] = s[i-1] - s[i]
// with the convention that s[-1] = s[n-1]-1
// this ensures ICR induces cyclic shift of
// difference array
vector<int> diff_array(const vector<int> &s) {
    vector<int> res(n);
    for(int i = 0; i < n; i++) {
        res[i] = s[(i-1+n)%n] - s[i] - (i==0);
        res[i] = (res[i] + base) % base;
    }
    return res;
}

// Incremented Cyclic Register rule
vector<int> icr(const vector<int> &s, int k) {
    vector<int> t = s;
    rotate(t.begin(), t.begin() + 1, t.end());
    t.back() = (t.back() + int(k))%base;
    return t;
}
\end{lstlisting}
\columnbreak
\begin{lstlisting}[language=C++]
// Checks if s is in Reps.
// This is true iff the difference array of s
// is positioned such that the lexicographically
// minimal rotation ends its first run of zeros
// at the last element.
// This ensures that applying ICR_0^-1 to s
// decreases the first element after the run.
bool reps(const vector<int> &s, int depth) {
    auto da = diff_array(s);
    if(da.back() == 0) return false;

    int i = n-1;
    while(i != 0 && da[i-1] == 0) i--;
    if(i == 0) return false;

    rotate(da.begin(), da.begin()+i, da.end());
    if(!min_lex(da)) return false;

    return (depth-s.back()) % base == 0;
}

// Generate the full de Bruijn sequence,
// with the transition function defined by
//    P(s) = ICR_0(s) if ICR_0(s) is in Reps
//           ICR_2(s) if ICR_1(s) is in Reps
//           ICR_1(s) otherwise
vector<int> generate() {
    vector<int> result;
    vector<int> s(n, 0);
    int depth = 0;
    do {
        auto a = icr(s, 0);
        auto b = icr(s, 1);
        auto c = icr(s, 2);
        if(reps(a, depth+1)) s = a, depth++;
        else if(reps(b, depth)) s = c, depth--;
        else s = b;
        result.push_back(s.back());
    } while(s != vector<int>(n, int(0)));
    return result;
}

int main() {
    cout << "enter n and base: " << flush;
    cin >> n >> base;
    auto result = generate();

    for(int i : result) cout << i;
    cout << endl;
}
\end{lstlisting}
\end{multicols}
\newpage

\bibliographystyle{plain}
\bibliography{optimal}
\medskip

{\small
\noindent
Nicol\'as Álvarez \\
  ICC CONICET Argentina -  {\tt  nico.alvarez@gmail.com}
\medskip

\noindent
Ver\'onica Becher \\
 Departamento de  Computaci\'on, Facultad de Ciencias Exactas y Naturales \& ICC  \\
 Universidad de Buenos Aires \&  CONICET  Argentina-  {\tt  vbecher@dc.uba.ar}
\medskip

\noindent
Martín Mereb \\
 Departamento de Matemática, Facultad de Ciencias Exactas y Naturales \& IMAS \\
 Universidad de Buenos Aires \&  CONICET Argentina-  {\tt  mmereb@gmail.com}
\medskip

\noindent
Ivo Pajor\\
 Departamento de Computaci\'on, Facultad de Ciencias Exactas y Naturales \\
 Universidad de Buenos Aires \&  Argentina-  {\tt  pajorivo@gmail.com}
\medskip

\noindent
Carlos Miguel Soto\\
 Departamento de Computaci\'on, Facultad de Ciencias Exactas y Naturales \\
 Universidad de Buenos Aires \&  Argentina-  {\tt  miguelsotocarlos@gmail.com}
}
\end{document}